\documentclass[technote, a4paper, onecolumn]{IEEEtran}  
\newcommand{\avtor}{Aleksandar Simevski}
\newcommand{\naslov}{Theoretical Aspects of a Design Method for Programmable NMR Voters}
\usepackage[pdftex]{graphicx}
\usepackage[caption=false]{subfig}
\usepackage{amssymb}
\usepackage{amsthm} 
\graphicspath{{H:/thesis/pictures/}}
\DeclareGraphicsExtensions{.jpg}
\usepackage[cmex10]{amsmath}
\newtheorem{property}{Property}

\newtheorem{corollary}{Corollary}

\newtheorem{ex}{Example}

\begin{document}
\title{\naslov} 

\author{\IEEEauthorblockN{Elena Hadzieva}
\IEEEauthorblockA{\\\textit{University of Information Science and Technology\\"St. Paul the Apostle"
\\Partizanska bb, 6000 Ohrid, Macedonia\\elena.hadzieva@uist.edu.mk}}

\and 
\vspace{0.05in}
\IEEEauthorblockN{\avtor}
\IEEEauthorblockA{\\\textit{IHP GmbH\\Im Technologiepark 25, D-15236 Frankfurt (Oder), Germany\\simevski@ihp-microelectronics.com}}


}


\maketitle

\begin{abstract} 

Almost all dependable systems use some form of redundancy in order to increase fault-tolerance. Very popular are the $N$-Modular Redundant (NMR) systems in which a majority voter chooses the voting output. However, elaborate systems require fault-tolerant voters which further give additional information besides the voting output, e.g., how many module outputs agree. Dynamically defining which set of inputs should be considered for voting is also crucial. Earlier we showed a practical implementation of programmable NMR voters that self-report the voting outcome and do self-checks. Our voter design method uses a binary matrix with specific properties that enable easy scaling of the design regarding the number of voter inputs N. Thus, an automated construction of NMR systems is possible, given the basic module and arbitrary redundancy $N$. In this paper we present the mathematical aspects of the method, i.e., we analyze the properties of the matrix that characterizes the method. We give the characteristic polynomials of the properly and erroneously built matrices in their explicit forms. We further give their eigenvalues and corresponding eigenvectors, which reveal a lot of useful information about the system. At the end, we give relations between the voter outputs and eigenpairs.

\end{abstract} 
\IEEEpeerreviewmaketitle

\section{Introduction}\label{sec_introduction} 

Widely used scheme for increasing system dependability is $N$-Modular Redundancy (NMR). Fig.~\ref{fig_nmr_sys} presents an NMR system. The $N$ identical modules $M_{0},M_{1},\dotsc,M_{N-1}$ fed with the same input $z$ are expected to produce equal outputs $x_{0},x_{1},\dotsc,x_{N-1}$. However, in a real system the modules are subject to faults that lead to differences in these outputs. Therefore, a decision maker D selects the final output of the system $y$. One of the most frequently used decision makers is the majority voter, where at least $\lfloor N/2+1 \rfloor$ outputs of the $N$ modules have to be equal.

\begin{figure}[!ht]
    \centerline{\includegraphics[width=0.5\linewidth]{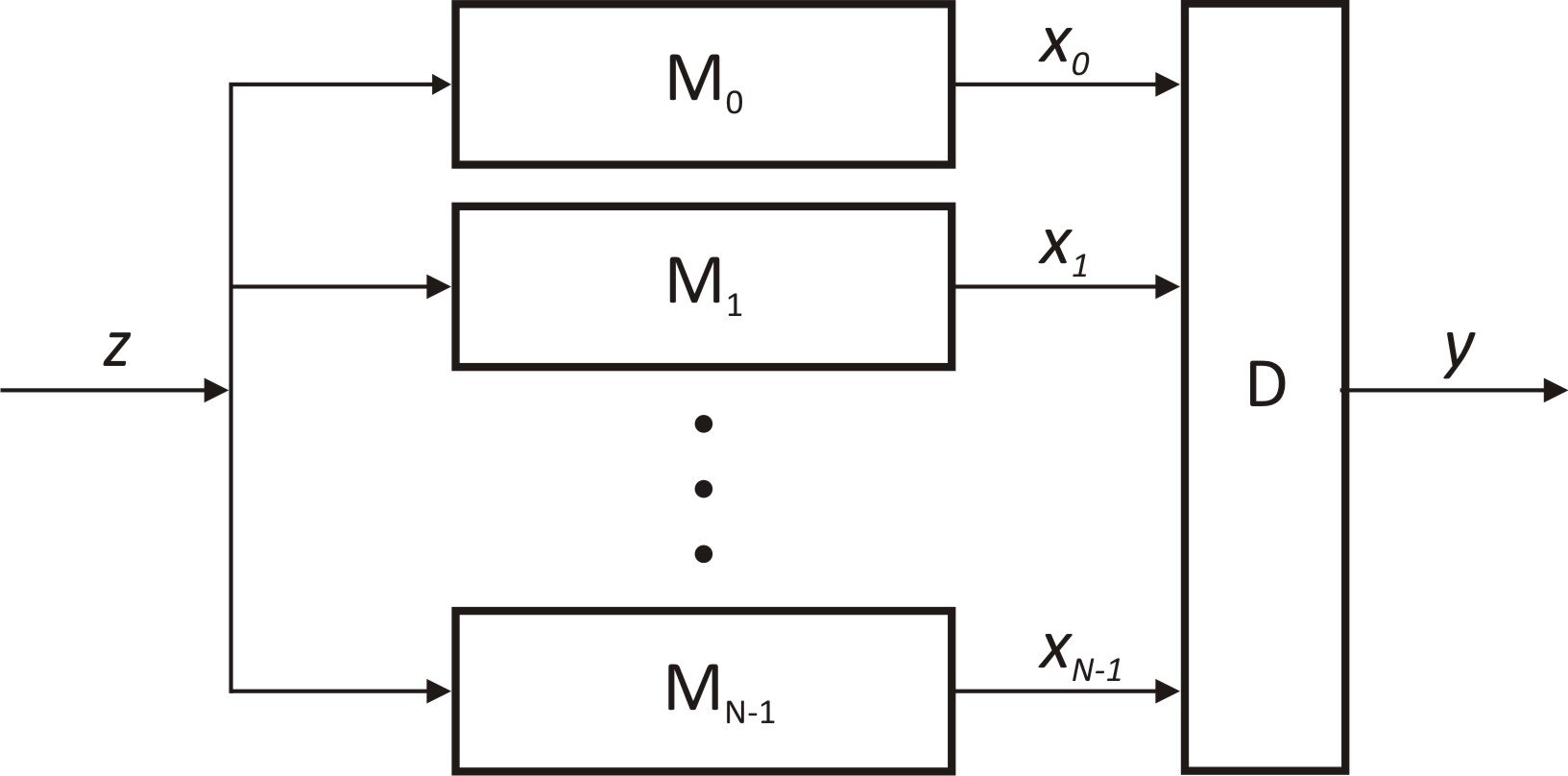}}
    \caption{NMR system}
    \label{fig_nmr_sys}
\end{figure}

Dependable systems employ some form of redundancy (time, space, information) which affect other system properties such as performance, power consumption or complexity (cost). A trade-off is therefore necessary. However, intelligent mechanisms may enable a dynamic trade-off, i.e., increase dependability and performance or lower power consumption on demand. Consider the dependable 4MR system depicted in Fig.~\ref{fig_dep_sys} as an example. The system acquires information by four identical sensors measuring the same physical quantity. This information is further processed by four processors that output the results $x_{0},x_{1},x_{2}$ and $x_{3}$.

\begin{figure}[!ht]
    \centerline{\includegraphics[width=0.5\linewidth]{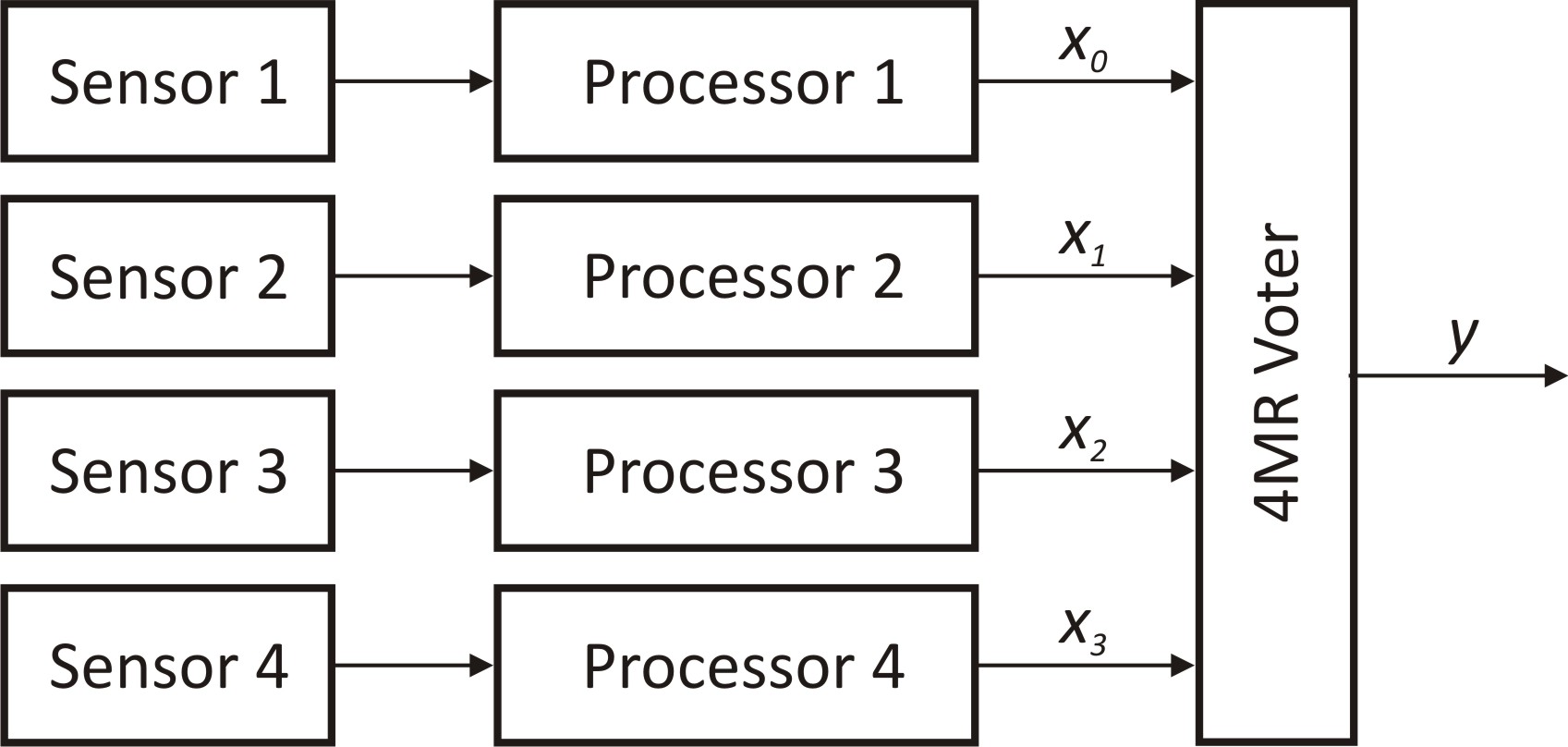}}
    \caption{Dependable 4MR system}
    \label{fig_dep_sys}
\end{figure}

By observing the responses of the processors over a period of time, the system could differentiate between permanent and transient faults in the processor-sensor pairs. Thus, if the system detects a permanent fault in one of the processor-sensor pairs, it may decide to switch them off in order to save power. Furthermore, consider the following NMR on demand (NMROD) adaptive behavior. Normally, only two processor-sensor pairs are operating in a dual-modular redundant (DMR) fashion. The power supply is switched off for all other pairs. As long as the two results are equal, the output of voting is equal to the results and the operation is considered error-free. A single disagreement between the two operating pairs is a signal for the system to power up a third pair and restart the operation in a triple-modular redundant (TMR) fashion. The fourth pair could be included only in critical situation when faults frequently occur, otherwise the system may opt  switching back to DMR. Besides the output of voting $y$, these example systems have to know exactly which processor-sensor pairs disagree, as well as the total number of pairs that disagree. Furthermore, they require dynamically building 1MR to NMR systems with any possible combination of processor-sensor pairs. In this discussion we have assumed that the voter itself is not subject to faults. However, system operation is compromised if a fault occurs in the voter. Therefore, it is preferable to have some dependability mechanisms which detect and report incorrect voter operation, or if possible, mask the errors.

So far we have illustrated our motivation for a special type of decision maker -- a programmable NMR voter with self-report and self-checking capabilities that is suitable for all the scenarios discussed previously. These voters describe the situation at their inputs, e.g., which modules disagree. Moreover, they could be dynamically programmed in order to form different NMR systems on the fly. In \cite{Simevski2012a} we show an intuitive method for designing such type of voters as well as the results from their actual implementation. Furthermore, we use these voters in order to investigate a dynamic scheme of core-level NMR in multiprocessors~\cite{Simevski2014a}. In this paper, we present the theoretical aspects of the method and we formally prove our assumptions. This is important since the method enables automated construction of elaborate NMR systems. That is, given the basic module (e.g., a single processor-sensor pair) and arbitrary redundancy $N$, the whole system could be built automatically. We present technical details of a register-transfer level NMR system generator in~\cite{Simevski2013a}.

The rest of the paper is organized as follows. Section~\ref{sec_related} presents related work. In Section~\ref{sec_definitions} we give a complete formal specification of our voters as well as some basic definitions that we use in the following Sections. We describe the method in Section~\ref{sec_design} and give its formal description and proofs of  properties in Section~\ref{sec_results}. The conclusion is in Section~\ref{sec_conclusion}.

\section{Related work}\label{sec_related}

A totally self-checking TMR system with concurrent error location capability is presented in~\cite{Jiang1999}. The system determines whether an error occurred during voting as well as its location. The error coverage is 100\%, i.e., the error can be detected in the redundant modules, the voter, or the error-checking circuit. The work is compared to a similar scheme proposed in~\cite{Gaitanis1988}. Yet another technique for increasing the reliability of NMR voters based on error correction by Alternate-Data Retry is introduced in~\cite{Takaesu2004}.

While the focus in~\cite{Jiang1999}, \cite{Gaitanis1988} and~\cite{Takaesu2004} is locating the error by using special circuits that observe the outputs of the redundant modules and the voter, our primary target is establishing a design method for programmable NMR voters which besides self-checks, output additional information for the state of their inputs. In particular, here we pay special attention to the mathematical analysis of this method in order to confirm its validity and importance, and enhance its capabilities.

Design of a reconfigurable NMR system is introduced in~\cite{Lo1990}. The design method enables scalability regarding the number of redundant modules $N$ and adaptability. Moreover, the authors in~\cite{Ruiz2008} present a strategy for automated generation of redundant modules and a corresponding majority voter. On the other side, the method that we present here enables not only simple but also elaborate NMR system generation (such as dynamic NMROD), using special NMR voters.

Dependability and performance analyses of NMR systems are given in~\cite{Srihari1982, Koren1979, Beaudry1978}, while dependability modeling of NMROD systems is found in~\cite{Al-Hashimi2001, Lombardi2001}.

\section{Basic definitions and voter specification}\label{sec_definitions}

An NMR system is practically determined by the properties and characteristics of the decision maker. As said, the most freqently used decision makers are various types of voters. We first give some basic definitions and make a short voter classification in order to set the frame for the following Sections. Then, we specify our type of voter.

Let the set of inputs of an NMR voter be $\mathcal{A}=\{x_0,x_1,\dotsc,x_{N-1}\}$. The absolute difference between the input values $x_i$ and $x_j$ wil be denoted by $\delta_{ij}$, i.e. $\delta_{ij}=\left|x_j-x_i\right|$. Exact voting algorithms consider $x_i$ and $x_j$ equal only if $\delta_{ij} = 0$, while inexact voting algorithms allow defining a $\sigma$ parameter and consider $x_i$ and $x_j$ equal if $\delta_{ij} < \sigma$. At last, approved voting algorithms define a set or range of approved input values. The voter considers $x_i$ and $x_j$ equal if they belong to the defined set/range. A complete voter classification with in-depth analyses is given in~\cite{Parhami1994}. Generally, the voters are marked by an $M$-out-of-$N$ label denoting that voting is successful if there are at least M equal inputs of the $N$ inputs in total. If $M\leq N/2$, ambiguous situations may occur since more than one input values could be legitimate candidates for the voting output.

Although in this paper we mainly assume an exact 1-out-of-$N$ voter, the design method is general and could be applied for almost any voter type. Fig.~\ref{fig_voter} depicts our voter which reports its state and checks its own operation.

\begin{figure}[!h]
    \centerline{\includegraphics[width=0.5\linewidth]{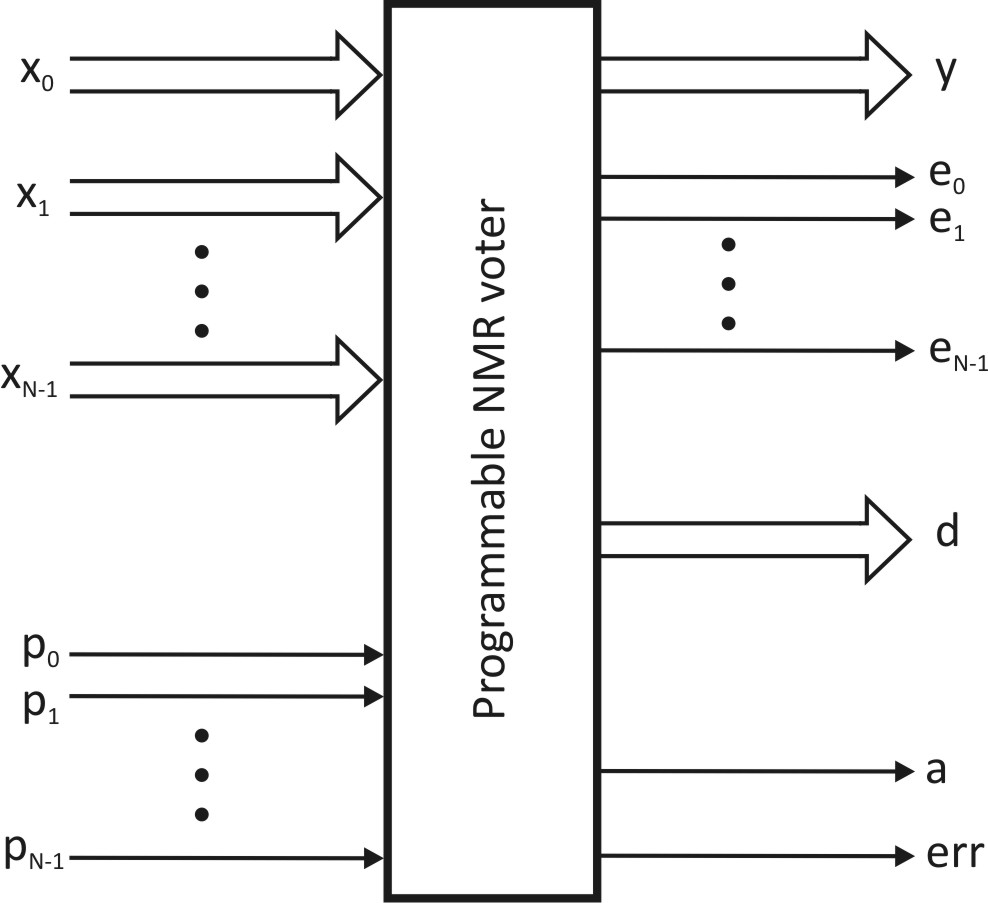}}
    \caption{Programmable NMR voter with self-report and self-checks}
    \label{fig_voter}
\end{figure}

The voting output $y$ is equal to $x_i$, where $x_i$ is in the largest group of equal inputs. The $d$ output gives the total number of inputs which differ from $y$. Equivalently, the voter could use an $eq$ output, which gives the total number of inputs that are equal to the output of voting $y$. Outputs $e_i,\quad i=0,1,\dots,N-1$ specify exactly which input $x_i$ equals to $y$, i.e. $e_i=1$, if $x_i=y$, and $e_i=0$ if $x_i\not=y$. The $a$ output signals ambiguous input situations where $y$ could be equal to any of the legitimate candidates for the voting outputs. The $err$ output signals an unsuccessful self-check. Actually, outputs $d$ (or $eq$), $e_i,\quad i=0,1,\dots,N-1$, and $a$ describe what is happening at the voter inputs. We refer to these outputs as the Input State Descriptor (ISD).

Furthermore, the voter is programmable. Each of the $x_i$ inputs could be dynamically programmed to be an \textit{active input}, through a special signal $p_i$. Each active input is included in the voting process, while all inactive inputs are excluded. This enables dynamically forming NMR systems (varying $N$) with any possible input combination. For example, defining $p_1=p_2=active$ and $p_0=p_3=inactive$ in a 4MR system, transforms the system to a 2MR system taking into account only modules 1 and 2; modules 0 and 3 do not participate in the voting process; later, however, module 3 may be included -- making a 3MR system. The programming signals imply two special configurations. Firstly, if no inputs are active, a 0MR system is formed, which is illegal. In this situation, the voter outputs are undefined. Secondly, if only one input is defined as active (1MR system), $y$ is always equal to the active input $x_i$. Thus, at least one input should be active for proper operation.

At last, but not least important is that the method enables scaling. That is, the complete voter with interface as in Fig.~\ref{fig_voter} could be generated solely by specifying the $N$ parameter. Furthermore, the design method is general in the sense that the specific implementation could be done either in hardware, software or any other technology. In~\cite{Simevski2012a} and~\cite{Simevski2013b} we show hardware and software realization, respectively.

\section{Voter design method}\label{sec_design}

Our method is based on a binary matrix that reflects the equal inputs of the voter. The matrix enables determining the voting output and the ISD, and performing self-checks. 

\subsection{Matrix construction}\label{subsec_matrix}

Inherently, the set of voter inputs $\mathcal{A}$ might contain repeatable elements. Let $\mathcal{A}$ contain $m\leq N$ different elements. If some element $x\in \mathcal{A}$ is repeated $k$ times in $\mathcal{A}$, then we say that \textit{the frequency of $x$ in $ \mathcal{A}$ is $ k$} (or simply, \textit{the frequency of $x$ is $k$}). Let $ f_1, f_2, \ldots, f_m$ are all possible frequencies of the elements of $\mathcal{A}$; then $f_1+f_2+\ldots+f_m=N.$

We construct the matrix $\mathbf{A}=[A_{ij}]_{N\times N}$, corresponding to the set $\mathcal{A}$, as follows:

$$A_{ij}=\left\{\begin{array}{ll}1,& x_i=x_j\\0, & x_i\not= x_j, \end{array}\right. \quad i,j=\overline{0,N-1}$$

(In the further text, we use the shorter notation $i=\overline{r_l,r_h}$ to express an integer index range from $r_l$ to $r_h$ with step 1. E.g., $i=\overline{0,N-1}$ instead of $i=0,1,\dots,N-1$.)

By  its definition, the matrix $\mathbf{A} $ is symmetric, with ones on its main diagonal.  If all voter inputs are different one from each other than $\mathbf{A}$ equals the identity matrix. In the opposite case, if all voter inputs are the same, than all matrix elements are ones. Additionally, the matrix $\mathbf{A}$ represents the relation ``='', defined on the set $\mathcal{A}$. As such, the matrix $\mathbf{A}$ represents equivalence  relation (that is, reflexive, symmetric and transitive).

\medskip
\begin{ex}\label{ex_matrix_build}
Let $N=4$, $x_0=20$, $x_1=30$, $x_2=20$, $x_3=10$, and all inputs are  active ($p_i=1$). The set $\mathcal{A}=\{20, 30, 20, 10\}$ has three different elements ($m=3$); their frequencies are $2,1,1$. The corresponding matrix would be: \[\mathbf{A}=\left[\begin{array}{cccc} 1&0&1&0\\0&1&0&0\\1&0&1&0\\0&0&0&1\end{array}\right]\,.\]
\end{ex}

\medskip

The voter has additional $N$ input signals that define which of its inputs  are active, thus dynamically setting the voter as 2MR, 3MR, \dots, NMR voter. The question is, how is this reflected into the $\mathbf{A}$ matrix. We consider each inactive input $x_i$ to be different from all other inputs $x_j, \, j\ne i$. Thus, $A_{ij}=A_{ji}=0, \forall j\ne i$ and $ A_{ii}=1$.

\subsection{Construction of ISD}\label{subsec_isd}

Taking into consideration that the matrix $\mathbf{A} $ is symmetric, with main diagonal of ones, all information about the input state can be obtained from the elements above (or below) its main diagonal  i.e., the elements $A_{ij},\, i=\overline{0,N-2},\, j=\overline{i+1,N-1}$. The elements of the matrix above the main diagonal from Example~\ref{ex_matrix_build} are: \[\begin{array}{ccc}0&1&0\\&0&0\\&&0\end{array}\]

By simply filling the missing $\displaystyle\frac{(N-1)(N-2)}{2}$ places with zeros (in the present case there are 6 such places, since $N=4$), we get a reduced $(N-1)\times(N-1)$ matrix: \[\mathbf{A}_R=\left[\begin{array}{ccc} 0&1&0\\0&0&0\\0&0&0\end{array}\right]\] where we preserve the enumeration of rows and columns ($i=\overline{0,N-2}$ and $j=\overline{i+1,N-1}$) as in the $\mathbf{A}$ matrix for the elements above the main diagonal.

Now the ISD (signals $d$, $e_i$ and $a$) could be simply determined as follows.
\begin{equation}\label{d}
d=\displaystyle\min_{i=0,1,\dotsc,N-2}|\{A_{ij}|A_{ij}=0,j=\overline{i+1,N-1}\}|,
\end{equation}
where the notation $|X|$ is used for the cardinality of a given set $X$. Actually, to find $d$, we search all (incomplete) rows, $i=\overline{0,N-2},\, j=\overline{i+1, N-1},\,$ of the $\mathbf{A}_R$ matrix, to find a row $i=I$, with the smallest number of zeros; then $d=|\{A_{Ij}|A_{Ij}=0,j=\overline{i+1,N-1}\}|$ and $eq=N-d$. We assign $y=x_{I}$, $e_{j}=1$ for $j=I$ and $e_j=0$, for $j<I$. Passing through columns $j=\overline{I+1,N-1}$ of row $I$ we determine $e_j$, for $j > I$, with $e_j=A_{Ij}$. The ambiguous signal $a$  is set to 1 if more than one (incomplete) row with the same, smallest number of zeros are encountered, otherwise $a=0$.

For instance, the smallest number of zeros in Example~\ref{ex_matrix_build} is in row $i=0$. Thus, $y=x_0=20$, $d=2$ (two zeros in row $i=0$), $eq=2$, $e_0=1$. Passing through row 0 we determine $e_j$ for $j > 0: e_1=A_{01}=0$, $e_2=A_{02}=1$, $e_3=A_{03}=0$, i.e., we distinguish which inputs are equal to the output of voting $y$. The ambiguous signal is zero ($a=0$) since we have a single row in the matrix with the smallest number of zeros.

Probing many examples (with many variations of input sets) we observed that $eq$ is exactly equal to the largest eigenvalue of the $\mathbf{A}$ matrix and that all other eigenvalues are integers. The question is, does this property hold in general, and, is it possible to prove it?

\subsection{Construction of self-checks}\label{subsec_self_checks}

In our previous paper \cite{Simevski2012a} self-check construction was based on violations of the transitivity property of the equivalence relation represented by the matrix $\mathbf{A}$. More precisely, it was based on the misrepresentation in the matrix $\mathbf{A}$ (consequently in the matrix $\mathbf{A}_R)$ of the following obvious property of the relation ``='' defined on the set of voter inputs.

\begin{equation}\label{eq_syllogism}
\mbox{if}\quad x_i=x_j\wedge x_j=x_k\quad\mbox{then}\quad x_i=x_k.\end{equation}
\begin{ex}\label{ex_self_checks}
For $N=4$, $x_0=x_1=x_2\neq x_3$ \[\mathbf{A}_R=\left[\begin{array}{ccc} 1&1&0\\0&1&0\\0&0&0\end{array}\right].\]

\noindent Suppose that $\mathbf{A}_{R12}=A_{12}$ is set to 0 instead of 1.  That is $A_{12}=A_{21}=0$, respectively. The matrices would become\[\mathbf{A}_R=\left[\begin{array}{ccc}1&1&0\\0&0&0\\0&0&0\end{array}\right],\quad\mathbf{A}=\left[\begin{array}{cccc}1&1&1&0\\1&1&0&0\\1&0&1&0\\0&0&0&1\end{array}\right].\]

\noindent Here, it is obvious that in the matrix representation of the method, the transitivity is violated. The $\mathbf{A}$ matrix simultaneously states that $x_0=x_1$, and $x_0=x_2$, but also that $x_1 \neq x_2$.
\end{ex}

\medskip
When the transitivity is violated as described in the example \ref{eq_syllogism}, we say that the matrix is \textit{erroneously built}. Erroneously built matrices indicate one or more such errors in voting. The voter could use these matrices to do self-checks. For instance, the voter could check if transitivity~(\ref{eq_syllogism}) is satisfied for each $i=\overline{0,N-2}$ of the $\mathbf{A}_R$ matrix, and each $j$ and $k$ where $j>i$ and $k>j$. If the self-check passes, then $err=0$ else $err=1$. Nevertheless, this simple check of transitivity violation does not mean that the voter is 100\% operating correctly. It only tells if an error is present in the matrix information or not. In other words, errors in the voter parts that later use the matrix information may not be caught without additional checks.

As we did for properly built matrices, here too, we examined the eigenvalues of erroneously built matrices. In this case, all experiments indicated that they always have a non-integer eigenvalue. Thus, a challenge to deeply analyze the basic matrix upon which our method is built, was posed. 


\section{Theoretical aspects}\label{sec_results}

\newcommand{\be}{\begin{equation}}
\newcommand{\ee}{\end{equation}}
\newcommand{\ds}{\displaystyle}


In this Section we give proofs of assertions that were stated intuitively in~\cite{Simevski2012a} and state and prove new assertions related to the matrix.The Section is divided into three Subsections which treat the properties of a properly and an erroneously built matrix, as well as the relations of the characteristics of these matrices with the voter outputs.

\subsection{Characteristics of a properly built matrix}\label{properly_matrix}
\medskip







Some of the obvious properties of the properly built  matrix $\mathbf{A}$  were stated right after its definition. Another straightforward property is that the matrix has only real eigenvalues, since it is real symmetric and therefore  Hermitian. We found that the eigenvalues and eigenvectors of the $\mathbf{A}$ matrix give a lot of information about the NMR voter.



\medskip
Recall that $ f_1, f_2,\ldots, f_m$ are all possible frequencies of the elements of  $ \mathcal{A}$. Generality is preserved if $ f_1 \geq f_2 \geq \ldots \geq f_m.$ The following property serves as a basis for deriving the next few properties.

\medskip
\begin{property}\label{prop_similar}
The $\mathbf{A}$ matrix is similar to the block -- matrix
$\overline{\mathbf{A}}=diag \{\mathbf{1}_{f_1}, \mathbf{1}_{f_2},
\ldots, \mathbf{1}_{f_m}\}$, where $\mathbf{1}_{f_i}$ is
$f_i\times f_i$ matrix whose elements are all ones.
\end{property}

\begin{proof}\label{proof_similar}
The $\mathbf{A}$ matrix is similar to the $\overline{\mathbf{A}}$
matrix, by the similarity  transformation $\mathbf{P}$, where
$\mathbf{P}$ is a product of a finite number of permutation matrices.
\end{proof}

Note that the matrix $\bar{\mathbf{A}}$ represents the same set of voter inputs, but with reordered elements. The elements are listed such that the set starts with the same elements with highest frequency, followed by the same elements with non-increasing  frequencies.  

\medskip
\begin{ex}\label{ex_similar}
For the matrix from Example~\ref{ex_matrix_build}, we  have:
\[\left[\begin{array}{cccc}1&1&0&0\\1&1&0&0\\0&0&1&0\\0&0&0&1\end{array}\right]=
\mathbf{P}_{12}^{-1}\cdot\left[\begin{array}{cccc}1&0&1&0\\0&1&0&0\\1&0&1&0\\0&0&0&1\end{array}\right]\cdot\mathbf{P}_{12},$$
where
$$\mathbf{P}_{12}=\mathbf{P}_{12}^{-1}=\left[\begin{array}{cccc}1&0&0&0\\0&0&1&0\\0&1&0&0\\0&0&0&1\end{array}\right]\,,\]
i.e.,
$$\overline{\mathbf{A}}=\mathbf{P}_{12}^{-1}\mathbf{A}\mathbf{P}_{12}.$$
Here, the permutation matrix $\mathbf{P}_{12}$ is obtained from
$4\times 4$ identity matrix, by interchanging 1-st and 2-nd row.
Multiplication $\mathbf{P}_{12}^{-1}\cdot \mathbf{A}$ interchanges
1-st and 2-nd rows of $\mathbf{A}$ and multiplication by
$\mathbf{P}_{12}$ interchanges 1-st and 2-nd column of
$\mathbf{P}_{12}^{-1}\cdot\mathbf{A}$. In this way, we get the
block -- matrix $\overline{\mathbf{A}}$, with blocks of
ones on the main diagonal, with sizes $2\times 2$, $1\times 1$ and
$1\times 1$. The sizes of the blocks are exactly equal to the frequencies of the elements of
$\mathcal{A}$.
\end{ex}

\medskip
It is easy to see (according to the Sylvester criterion) that $\overline{\mathbf{A}}$ is positive semi-definite. So is the $\mathbf{A}$ matrix~\cite{Meyer2000}, which implies that their eigenvalues are non-negative. So far, we know that the eigenvalues of $\mathbf{A}$ are non-negative reals. The following properties reveal, step-by step, the whole spectrum of $\mathbf{A}$.

%

\medskip
Let $\sigma(\mathbf{A})$ and $\rho(\mathbf{A})$ denote the spectrum and spectral radius of $\mathbf{A}$, respectively.

\medskip
\begin{property}\label{prop_nula} $0\in\sigma(\mathbf{A})$.\end{property}

\begin{proof}\label{proof_nula}
Two similar matrices have the same determinant, thus
$$\det\mathbf{A} =\det\overline{\mathbf{A}}=\det\mathbf{1}_{f_1}\cdot\det\mathbf{1}_{f_2}\cdot\ldots\cdot\det\mathbf{1}_{f_m}=0.$$
The relation $\det\mathbf{A}=0$ implies that $0$ is an eigenvalue
of $\mathbf{A}$.
\end{proof}

\medskip
\begin{property}\label{theo_max}
$\rho (\mathbf{A})=f_1$. Moreover, $\rho(\mathbf{A})\in \sigma (\mathbf{A})$, i.e., the largest frequency of the elements of $ \mathcal{A}$ is an eigenvalue of $\mathbf{A}$.
\end{property}

\begin{proof}\label{proof_max}
(This proof is different than the one given in~\cite{Simevski2012a}.) As concluded before, all of the eigenvalues of $\mathbf{A}$, $\lambda_i $, $ 1\leq i\leq N$ are non-negative, so its spectral radius \be\label{ro}\rho(\mathbf{A})=\displaystyle\max_{i=1,2,\ldots,N}|\lambda_i|=\max_{i=1,2,\ldots, N}\lambda_i.\ee
The inequality \be\rho(\mathbf{A})\leq ||\mathbf{A}||\ee holds for every norm of $\mathbf{A}$ (\cite{Meyer2000}, p.~497). If we choose $||\cdot ||_1$ -- norm, then we obtain \be\label{eq_pomalo}\rho(\mathbf{A})\leq ||\mathbf{A}||_1=\max_i\sum_j|A_{ij}|=f_{1},\ee since the largest absolute row sum in $\mathbf{A}$ is $f_1$.

On the other hand, $f_1$ is an eigenvalue of $\mathbf{A}$, since for the $N$ -- dimensional  vector $$\mathbf{y}_1=[\underbrace{1\,\,1\,\,\ldots\,\,1}_{f_1}\,\,0\,\,0\,\,\ldots\,\,0]$$ the equality $$\overline{\mathbf{A}}\mathbf{y}_1=f_1\mathbf{y}_1$$ is satisfied. (We use the fact that the similar matrices $\mathbf{A}$ and $\overline{\mathbf{A}}$ have  the same eigenvalues.)

Taking into account the fact that $f_1$ is an eigenvalue of $\mathbf{A}$, and the relations (\ref{ro}) and (\ref{eq_pomalo}), we conclude  that the maximal eigenvalue of $\mathbf{A}$ is $f_1=\rho(\mathbf{A})$.
\end{proof}


\medskip
\begin{property}\label{theo_fr}
All frequencies of the elements of $\mathcal{A}$ are eigenvalues of the $\mathbf{A}$ matrix.
\end{property}

\begin{proof}\label{proof_fr}
It is enough to show that $f_2, f_3,\ldots, f_m$ are eigenvalues of the $\overline{\mathbf{A}}$ matrix.

For $i=2,3,\ldots, m$, we define $N$-dimensional vectors $\mathbf{y}_i$ with $$ \mathbf{y}_i = [\underbrace{0\,\,0\,\,0\,\ldots\,0\,\,0\,\,0}_{f_1+f_2+\ldots+f_{i-1}}\,\,\underbrace{1\,\,1\,\ldots\,1\,\,1}_{f_i}\,\,0\,\,\ldots\,0]^T.$$

It is easy to check that $$\overline{\mathbf{A}}\mathbf{y}_i = f_i\mathbf{y}_i,\quad i=\overline{2,m}.$$
\end{proof}

The proofs of Properties~\ref{theo_max} and~\ref{theo_fr} contain
explicit  formulas of the eigenvectors of $\overline{\mathbf{A}}$,
that correspond to the eigenvalues of $ \overline{\mathbf{A}}$,
i.e., the eigenpairs of $\overline{\mathbf{A}}$ are $
(f_i,\mathbf{y}_i), i=\overline{1,m}$. What about the
eigenpairs of $\mathbf{A}$? Of course, the eigenvalues are the
same. We denote the eigenvectors of $\mathbf{A}$ corresponding to the
eigenvalues $f_1, f_2,\ldots, f_m$, by
$\mathbf{y}^{f_i}$ and give their description in the following property. 

\medskip
\begin{property}\label{theo_pos}
The eigenvectors of $ \mathbf{A}$, $\mathbf{y}^{f_i}$, are 0-1 $N$-dimensional vectors with ones at
the positions that coincide with the positions of the elements
with frequency $f_i$ of the set $ \mathcal{A}$.
\end{property}

\begin{proof}\label{proof_pos}
It can be easily checked that
$$\mathbf{A}\mathbf{y}^{f_i}=f_i\mathbf{y}^{f_i},$$ i.e., both
vectors, $\mathbf{A}\mathbf{y}^{f_i}$ and $f_i\mathbf{y}^{f_i},$
have $f_i$-s at the positions of the elements with frequency $f_i$.
All other components are zeros.
\end{proof}

\medskip
For the sake of clarity, we give the eigenvector $\mathbf{y}^{f_1}$ in its explicit form. Since $f_1$ is a frequency of some element  of
$\mathcal{A}$, there exist $k_1, k_2, \ldots, k_{f_1}
\in\{0,1,\ldots, N-1\}$, $ k_1<k_2<\ldots<k_{f_1}$ such that $$
x_{k_1}=x_{k_2}=\ldots=x_{k_{f_1}}.$$ The components of the vector
$\mathbf{y}^{f_1}$ are $$ y^{f_1}_{i}=\left\{\begin{array}{rl}1, &
i= k_1, k_2,\ldots, k_{f_1}\\0, & \hbox{\rm
otherwise}\end{array}\right.\,.$$ and
the components of the vector $\mathbf{A}\mathbf{y}^{f_1}=f_1
\mathbf{y}^{f_1}$ are $$
[\mathbf{A}\mathbf{y}^{f_1}]_i=\left\{\begin{array}{rl}f_1, & i=
k_1, k_2,\ldots, k_{f_1}\\0, & \hbox{\rm
otherwise}\end{array}\right.\,. $$ In
general, the eigenvector $ \mathbf{y}^{f_i}$ gives information
about the ordinal numbers of the  elements (inputs) that have frequency
$f_i$. At those positions $\mathbf{y}^{f_i}$ has ones. Other
elements of $\mathbf{y}^{f_i}$ are zeros.

{\it Remark:} Note that the rows (columns) of $ \mathbf{A}$ are its eigenvectors.


\medskip
\begin{ex}
For the matrix from the Example~\ref{ex_matrix_build}, 
$f_1=2$, $f_2=f_3=1$, $ \rho(\mathbf{A})=2=f_1.$  The eigenpairs
of $\overline{\mathbf{A}}$ for the frequencies $ f_1, f_2, f_3$
are: $(2, \mathbf{y}_1=[1\,\,1\,\,0\,\,0]^T),$
$(1,\mathbf{y}_2=[0\,\,0\,\,1\,\,0]^T),$ and
$(1,\mathbf{y}_3=[0\,\,0\,\,0\,\,1]^T).$ The eigenpairs of
$\mathbf{A}$ for the frequencies $ f_1, f_2, f_3$ are:
$(2,\mathbf{y}^{f_1}=[1\,\,0\,\,1\,\,0]^T),$
$(1,\mathbf{y}^{f_2}=[0\,\,1\,\,0\,\,0]^T),  $ and
$(1,\mathbf{y}^{f_3}=[0\,\,0\,\,0\,\,1]^T),$ from which we read
the information that $x_0=x_2$ is an element with frequency
$f_1=2$, $x_1$ and $x_3$ are elements with frequencies
$f_2=f_3=1$.
\end{ex}

\medskip
So far we showed that all frequencies of the elements of $\mathcal{A}$ and zero are eigenvalues  of $ \mathbf{A}$. The question is, whether $\mathbf{A}$ has some other eigenvalues?  Before answering this question (the answer is actually the property \ref{theo_spektar}) we will compute the $n\times n$ determinants  $ D_n(s)$ and $F_n(s)$ defined for $n\in \mathbb{N}$ and $s\in \mathbb{R},$ by \be D_n(s)=\left|\begin{array}{rrrrrr}s&-1&-1&-1&\ldots&-1\\
-1&s&-1&-1&\dots&-1\\-1&-1&s&-1&\ldots&-1\\& & & &\ddots&
\\-1&-1&-1&-1&\ldots&s
\end{array}\right|\qquad \hbox{\rm and} \ee \be F_n(s)=\left|\begin{array}{rrrrrr}-1&-1&-1&-1&\ldots&-1\\
-1&s&-1&-1&\dots&-1\\-1&-1&s&-1&\ldots&-1\\& & & &\ddots&
\\-1&-1&-1&-1&\ldots&s
\end{array}\right|.\ee
We claim that
\be\label{eq_matind} D_n(s)=(s+1)^{n-1}(s-n+1),\ee
\be\label{eq_matind1} F_n(s)=-(s+1)^{n-1}.\ee

By cofactor expansion of both determinants along their first row, we get:
\be\label{dn} D_n(s)=sD_{n-1}(s)+(n-1)F_{n-1}(s),\ee\be\label{fn} F_n(s)=-D_{n-1}(s)+(n-1)F_{n-1}(s). \ee 
For $ n=1$, $ D_1=s$ and $ F_1=-1$, which corresponds to the formulas. Assuming that (\ref{eq_matind}) and (\ref{eq_matind1}) hold for $ n\in \mathbb{N}$ and taking into account (\ref{dn}) and (\ref{fn}), we obtain:
\begin{eqnarray*}
D_{n+1}(s)&=&sD_n(s)+nF_n(s)=\\ &=&s(s+1)^{n-1}(s-n+1)-n(s+1)^{n-1}=\\&=&(s+1)^n(s-n),\\
F_{n+1}(s)&=&-D_n(s)+nF_n(s)=\\
&=&-(s+1)^{n-1}(s-n+1)-n(s+1)^{n-1}=\\&=&-(s+1)^n,
\end{eqnarray*}
which, by means of mathematical induction, proves the formulas  (\ref{eq_matind}) and (\ref{eq_matind1}).

\medskip
\begin{property}\label{theo_spektar}$\sigma (\mathbf{A})=\{0, f_1, f_2, \ldots, f_m\}$, i.e., the spectrum of $\mathbf{A} $ matrix consists of $ 0$ and
all frequencies of the elements of $ \mathcal{A}$.\end{property}

\begin{proof}\label{proof_spektar}
We use the notation $(\lambda \mathbf{I}-\mathbf{1})_{f_i} $ to indicate  $f_i\times f_i$ matrix
\be\label{eq_blok} (\lambda \mathbf{I}-
\mathbf{1})_{f_i}=\left[\begin{array}{rrrrrr}\lambda -1&-1&-1&-1&\ldots&-1\\
-1&\lambda-1&-1&-1&\dots&-1\\-1&-1&\lambda -1&-1&\ldots&-1\\& & &
&\ddots&
\\-1&-1&-1&-1&\ldots&\lambda -1
\end{array}\right].\ee

The $\lambda \mathbf{I}-\overline{\mathbf{A}}$ matrix is a block-diagonal matrix,  consisting of $m$ blocks of sizes $f_i\times f_i$ ($ i=\overline{1,m}$), of type~(\ref{eq_blok}). We give the characteristic polynomial of the $\overline{\mathbf{A}}$, i.e. $\mathbf{A}$ matrix,  in its explicit form, using  the formulas  (\ref{eq_matind}) and (\ref{eq_matind1}) (substituting $s=\lambda-1$, $n=f_i,\, i=\overline{1,m}$):
\begin{eqnarray*}
&&c_{\mathbf{A}}(\lambda)= \det (\lambda\mathbf{I}-\overline{\mathbf{A}})=\\
& &= \det(\lambda \mathbf{I}- \mathbf{1})_{f_1}\cdot \det(\lambda\mathbf{I}- \mathbf{1})_{f_2}\cdot \ldots \cdot \det(\lambda
\mathbf{I} - \mathbf{1})_{f_m}=\\
&&= D_{f_1}(\lambda-1)\cdot D_{f_2}(\lambda-1)\cdot\ldots\cdot
D_{f_m}(\lambda-1)=\\ &&=\lambda^{f_1-1}(\lambda -
f_1)\cdot\lambda^{f_2-1}(\lambda - f_2)\cdot\ldots\cdot\lambda^{f_m-1}(\lambda-f_m)=\\
&&=\lambda^{N-m}(\lambda-f_1)(\lambda-f_2)\ldots(\lambda-f_m).
\end{eqnarray*}
\end{proof}
The explicit form of the characteristic polynomial of the properly built matrix does not only give information about the spectrum of the $\mathbf{A}$ matrix, but also for the algebraic multiplicity of each eigenvalue. The importance of this fact will be elaborated later.

\subsection{Characteristics of an erroneously built matrix}\label{erroneously_matrix}
\medskip

In order to find the characteristic polynomial of the erroneously built matrix, we first define the $ n\times n$ determinant $ Q_n(s)$  (where $n\in \mathbb{N}$ and $s\in \mathbb{R}$) by:
\be Q_n(s)=\left|\begin{array}{rrrrrr}s&0&-1&-1&\ldots&-1\\
0&s&-1&-1&\dots&-1\\-1&-1&s&-1&\ldots&-1\\& & & &\ddots&
\\-1&-1&-1&-1&\ldots&s
\end{array}\right|.\ee Its value can be easily obtained by cofactor expansion along its first row, 
$$Q_n(s)=sD_{n-1}(s)+(n-2)s F_{n-2}(s), \quad\hbox{\rm or} $$
\be\label{eq_q} Q_n(s)=s(s+1)^{n-3}(s^2+(-n+3)s-n+4).\ee


\medskip
\begin{property}\label{theo_nontransitivity}
Let there exist three equal elements $x_i=x_j=x_k\in\mathcal{A}.$ If the following holds for the entries of the $\mathbf{A}$ matrix:
\be\label{eq_nontr}
A_{ij}=A_{ji}=1=A_{ik}=A_{ki}\,\wedge\, A_{jk}=0=A_{kj},\ee then it has a characteristic polynomial of the type
\begin{eqnarray*}
c_{\mathbf{A}}(\lambda)&=&\lambda^{N-m-2}(\lambda-1)(\lambda^2+(-f_l+1)\lambda-f_l+2)\cdot\\&
&\cdot(\lambda-f_1)\ldots(\lambda-f_{l-1})(\lambda-f_{l+1})\ldots(\lambda-f_m).
\end{eqnarray*}
\end{property}

\begin{proof}\label{proof_nontransitivity}
Note that, since there should exist at least three equal voter inputs to consider transitivity at all, the frequency $f_l$ should be greater or equal to three. Let the elements $ x_i=x_j=x_k$ have the frequency $f_l, \,\,f_l\in
\{f_1, f_2,\dots,f_m\}$. Let the  $\mathbf{A}$ matrix be erroneously
built, as described by (\ref{eq_nontr}). Then, there
exists a $\mathbf{P}$ matrix such that
$\overline{\mathbf{A}}=\mathbf{P}^{-1}\cdot
\mathbf{A}\cdot\mathbf{P},$ where $\mathbf{P}$ is a product of a finite number of permutation matrices, $
\overline{\mathbf{A}}=diag\{\mathbf{1}_{f_1}, \mathbf{1}_{f_2},
\ldots, \mathbf{1}_{f_{l-1}}, \mathbf{Q}, \mathbf{1}_{f_{l+1}},
\ldots, \mathbf{1}_{f_{m}}\}$ and $\mathbf{Q}$ is the $f_l\times
f_l$ matrix
$$\mathbf{Q}=\left[\begin{array}{rrrrrr}1&0&1&1&\ldots&1\\
0&1&1&1&\dots&1\\1&1&1&1&\ldots&1\\& & & &\ddots&
\\1&1&1&1&\ldots&1
\end{array}\right]\,.$$  Then the matrix $\lambda
\mathbf{I}-\overline{\mathbf{A}}$ is a block-diagonal matrix
consisting  of blocks of type (\ref{eq_blok}) for all $ f_i\in
\{f_1, f_2, \ldots, f_m\}\setminus\{f_l\}$, and the $f_l\times
f_l$ block $\lambda\mathbf{I}-\mathbf{Q}$. Thus, the
characteristic polynomial of $ \mathbf{A}$ is:
\begin{eqnarray*}
& &c_{\mathbf{A}}(\lambda)= \det (\lambda
\mathbf{I}-\overline{\mathbf{A}})=\\
&&= \det(\lambda \mathbf{I}- \mathbf{1})_{f_1}\cdot \det(\lambda
\mathbf{I}- \mathbf{1})_{f_2}\cdot \ldots \cdot \det(\lambda
\mathbf{I} - \mathbf{1})_{f_{l-1}}\cdot\\ & &\cdot  \det
(\lambda\mathbf{I}-\mathbf{Q})\cdot \det(\lambda \mathbf{I} -
\mathbf{1})_{f_{l+1}}\cdot \ldots \cdot \det(\lambda \mathbf{I} -
\mathbf{1})_{f_{m}}=\\&&= D_{f_1}(\lambda-1)\cdot
D_{f_2}(\lambda-1)\cdot\ldots \cdot
D_{f_{l-1}}(\lambda-1)\cdot\\&&\cdot Q_{f_l}(\lambda-1)\cdot
D_{f_{l+1}}(\lambda-1)\cdot\ldots \cdot
D_{f_{m}}(\lambda-1).
\end{eqnarray*}
Using (\ref{eq_matind}) and (\ref{eq_q}) (with substitutions $ s=\lambda-1$, and $ n=f_i,\,$ $i=\overline{1,m}$) we obtain:
\begin{eqnarray*}
c_{\mathbf{A}}(\lambda)&=&\lambda^{N-m-2}(\lambda-1)(\lambda^2+(-f_l+1)\lambda-f_l+2)\cdot\\&
&\cdot(\lambda-f_1)\ldots(\lambda-f_{l-1})(\lambda-f_{l+1})\ldots(\lambda-f_m).\end{eqnarray*}
\end{proof}

\medskip
\begin{corollary}\label{cor_notnat} If the $\mathbf{A}$ matrix is erroneously built, then it has two non-integer eigenvalues.\end{corollary}

\begin{proof}
The roots of the characteristic polynomial are  1, all frequencies except $f_l$,  then 0 (if $N-m>2$), and the scalars $$\lambda_{1/2}=\frac{f_l-1\pm\sqrt{f_l^2+2f_l-7}}{2}.$$ The last eigenvalues are non-integers, since $\sqrt{f_l^2+2f_l-7}$ is non-integer for $f_l\geq3$. We certify this with the inequality $$f_l^2<f_l^2+2f_l-7<(f_l+1)^2,$$ that holds for $f_l\geq 4,$ and $ f_l^2+2f_l-7=8,$ for $ f_l=3.$
\end{proof}

\medskip
Corollary~\ref{cor_notnat}  shows another way to  the voter how to do self-checks. Another useful fact in this direction is that the zero eigenvalue has algebraic multiplicity $N-m-2$ for an erroneously built matrix, opposed to the algebraic multiplicity $N-m$ for a properly built matrix. Similarly, for a properly built matrix, the eigenvalue 1 has algebraic multiplicity equal to the number of inputs with frequency 1 (including inactive inputs). For an erroneously built matrix, the algebraic multiplicity of 1 is bigger than this number for 1.


\medskip
\begin{ex}\label{pr5}
If $x_0=x_1=x_3\not=x_2$ ($N=4, m=2$, $ f_1=3, f_2=1$) and the corresponding matrix $\mathbf{B} $ is erroneously built, $$
\mathbf{B}=\left[\begin{array}{cccc}
1&1&0&0\\1&1&0&1\\0&0&1&0\\0&1&0&1\end{array}\right] $$ (the matrix implies that $ x_0=x_1, x_1=x_3, x_0\not=x_3$), then
$$\overline{\mathbf{B}}=\left[\begin{array}{cccc}
1&0&1&0\\0&1&1&0\\1&1&1&0\\0&0&0&1\end{array}\right]=
\mathbf{P}_{23}^{-1}\mathbf{P}_{13}^{-1}\mathbf{A}\mathbf{P}_{13}\mathbf{P}_{23}$$
and $$c_{\mathbf{B}}(\lambda)=(\lambda-1)^2(\lambda^2-2\lambda-1).$$
The eigenvalues are $ \lambda_1=1$ (since $ f_2=1$),  $\lambda_2=1$ (since $ 1$ is always an eigenvalue of an erroneously built matrix, see Property~\ref{theo_nontransitivity}) and two non-integer eigenvalues $ \lambda_{3/4}=1\pm\sqrt2.$
\end{ex}

\medskip
\begin{ex}
If $ x_0=x_1=x_2=x_3$ ($N=4, m=1$, $ f_1=4$) and the corresponding matrix $\mathbf{B}$ is erroneously built,
$$\mathbf{B}=\left[\begin{array}{cccc}
1&1&1&0\\1&1&1&1\\1&1&1&1\\0&1&1&1\end{array}\right] $$ (the matrix implies that $ x_0= x_2, x_2=x_3, x_0\not=x_3$), then
$$\overline{\mathbf{B}}=\left[\begin{array}{cccc}
1&0&1&1\\0&1&1&1\\1&1&1&1\\1&1&1&1\end{array}\right]=
\mathbf{P}_{13}^{-1}\mathbf{A}\mathbf{P}_{13}$$ and $$
c_{\mathbf{B}}(\lambda)=\lambda(\lambda-1)(\lambda^2-3\lambda-2).$$
The eigenvalues are $ \lambda_1=1$, with algebraic multiplicity 1; $\lambda_2=0$, with algebraic multiplicity $ N-m-2=4-1-2=1$ and two simple non-integer eigenvalues $
\lambda_{3/4}=\frac{3\pm\sqrt{17}}2.$
\end{ex}

\subsection{Matrix -- voter outputs relationship}\label{relations}
\medskip
At the end, we give the relations between the voter outputs (if the input set is $\mathcal{A}$) and the scalar characteristic of a properly built matrix $\mathbf{A}$ corresponding to the set $\mathcal{A}$. The output $y$ is actually the element $x_{k_1}=x_{k_2}=\ldots = x_{k_{f_{1}}}$ (recall the comment after Property ~\ref{theo_pos}); $eq=f_1$ -- the largest eigenvalue of $\mathbf{A}$; $ d=N-f_1$; $ e_i$ are the components of the vector $\mathbf{y}^{f_1}$; the ambiguous signal $ a=0$, if $f_1$ is simple eigenvalue and $a=1$ if $alg\,mult_{\mathbf{A}}(f_1)>1.$ $err=1$ if there is a non-integer eigenvalue of $\mathbf{A}$, and $err=0$ if all eigenvalues are non-negative integers.

In other words, the eigenvalues of $\mathbf{A}$ answer the questions like ``What are the frequencies of the inputs?'', ``What is the output of the voter?'' or ``Is the matrix erroneously built?''.  The corresponding eigenvectors answer the question ``What are the positions of the equal inputs (with the corresponding frequency)?''. The multiplicity of the largest eigenvalue answers the question ``Is there an ambiguity between the inputs?". Examples~\ref{ex7} and \ref{ex8} illustrate these issues.

\medskip
\begin{ex} \label{ex7} For the $\mathbf{B}$ matrix from Example~\ref{pr5} (because of its non-integer eigenvalues), the value of the error signal is $1$, $ err=1.$ If it was properly built,
$$\mathbf{B}=\left[\begin{array}{cccc}
1&1&0&1\\1&1&0&1\\0&0&1&0\\1&1&0&1\end{array}\right], $$
then $$ c_{\mathbf{B}}(\lambda)=\lambda^2(\lambda-1)(\lambda-3)$$
and the eigenvectors of $\mathbf{B}$ matrix corresponding to $f_1=3$ and $f_2=1$, are $\mathbf{y}^{f_1}=[1\,1\,0\,1] $ and $\mathbf{y}^{f_2}=[0\,0\,1\,0].$
We obtain: \begin{eqnarray*}eq &=& f_1=3;\\
d& =&N-f_1=1;\\
y&=&x_0=x_1=x_3\,\,(\hbox{we read from}\,\,\mathbf{y}^{f_1});\\
e_0&=&1=e_1=e_3;\\ e_2&=&0;\\
a&=&0,\,\,\,\hbox{since}\,\,\, alg\,mult_{\mathbf{A}}(f_1)=1;\\
err&=&0.\end{eqnarray*}
\end{ex}

\medskip
\begin{ex} \label{ex8} If $ x_0=20=x_2$, $ x_1=x_3=30$, then the corresponding matrix  is
 $$
\mathbf{D}=\left[\begin{array}{cccc}
1&0&1&0\\0&1&0&1\\1&0&1&0\\0&1&0&1\end{array}\right]. $$
Its characteristic polynomial is
$$ c_{\mathbf{D}}(\lambda)=\lambda^2(\lambda-2)^2,$$ which means that its eigenvalues are $\lambda_1=2$
and $\lambda_2=0 $. The eigenvectors corresponding to the
frequency $ f_1=f_2=2$ are $ [1\,0\,1\,0]$ and $ [0\,1\,0\,1]$.
We obtain: \begin{eqnarray*}eq &=& f_1=2;\\
d& =&N-f_1=2;\\ a&=&1,\,\,\,\hbox{since}\,\,\, alg\,mult_{\mathbf{A}}(f_1)=2;\\
y&=&x_0=x_1\,\,\, \hbox{or}\,\,\,y=x_1 =x_3;\\
e_0&=&1=e_2,\,\, e_1=0=e_3,\,\, \hbox{or},\\ e_0&=&0=e_2,\,\, e_1=1=e_3;\\
 err&=&0.\end{eqnarray*}
\end{ex}

\section{Conclusion} \label{sec_conclusion}

Outlining our motivation in Section~\ref{sec_introduction} we gave several examples of sophisticated dependable NMR systems. They actually led us to a design method for programmable NMR voters that self-report their state and self-check their operation. The method is based on a binary matrix, which enables simplicity and scalability of the voter design. We got experimental results that foreshadowed interesting matrix properties, which in this paper were shown to be true by rigorous mathematical proofs.  We characterized the design method through the most important matrix characteristics -- the eigenvalues and eigenvectors. All exposed, theoretically-proven characteristics of the method enhance its possibilities in different applications. Although in hardware-realized NMR systems $N$ is usually in the range from two to eight, in this paper we showed that the method is general and can be used to construct NMR systems for any natural number $N$. A software realization for large $N$ is given in~\cite{Simevski2013b}.

\end{document}